\documentclass[10pt,english,onecolumn]{llncs}
  \author{Edoardo Persichetti}%{-0.2cm}}
  \institute{Florida Atlantic University \vspace{-0.6cm}}

\title{\bf Code-based Key Encapsulation\\ 
from McEliece's Cryptosystem}
\usepackage{babel}
\usepackage{amsfonts}
\usepackage{amsmath}
\usepackage{amssymb}
\usepackage{theorem}
\usepackage[latin1]{inputenc}
\usepackage[small,bf]{caption}
\usepackage{multirow}
\usepackage{mathtools}
\usepackage{enumerate}
\usepackage{array}
\usepackage{bigstrut}
\usepackage{exscale,relsize}
\usepackage{graphicx}
\newcommand{\4}{\noindent}
\newcommand{\5}{\bigskip}

\newcommand{\Ts}{\textsf}
\newcommand{\Gen}{\textsf{KeyGen}}
\newcommand{\Enc}{\textsf{Enc}}
\newcommand{\Dec}{\textsf{Dec}}

\newcommand{\rand}{\xleftarrow{\$}}
\newcommand{\pk}{\textsf{pk}}
\newcommand{\sk}{\textsf{sk}}

\newcommand{\pr}{\textsf{Pr}}
\newcommand{\Adv}{\textsf{Adv}}

\newcommand{\KDF}{\textsf{KDF}}
\newcommand{\KEM}{\textsf{KEM}}
\newcommand{\DEM}{\textsf{DEM}}
\newcommand{\HY}{\textsf{HY}}

  \newtheorem{assumption}{Assumption} %corsivo o no?
\begin{document}
\date{}
\maketitle

\begin{abstract}In this paper we show that it is possible to extend the framework of Persichetti's Nierreiter-based KEM~\cite{Edoardo-2013} and create a secure KEM based on the McEliece protocol. This provides greater flexibility in the application of coding theory as a basis for cryptographic purposes.
%Cryptographic schemes based on coding theory are one of the most accredited choices for cryptography in a post-quantum scenario. In this work, we present a hybrid construction based on the Niederreiter framework that provides IND-CCA security in the random oracle model. In addition, the construction satisfies the IK-CCA notion of anonymity whose importance is ever growing in the cryptographic community.
\end{abstract}

\section{Introduction}
A \emph{Hybrid Encryption} scheme is a cryptographic protocol that uses public-key encryption as means to securely exchange a key, while delegating the task of encrypting the body of the message to a symmetric scheme. The public-key component is known as \emph{Key Encapsulation Mechanism (KEM)}. The first code-based KEM, utilizing the Niederreiter framework~\cite{nied}, was presented by Persichetti in~\cite{Edoardo-2013} and successively implemented in \cite{mcbits}. In this paper, we expand on Persichetti's work and prove that if we use the McEliece approach~\cite{mceliece} we are still able to obtain a secure KEM. This is a novel construction, with a great potential impact, especially considering NIST's recent call for papers for secure post-quantum primitives~\cite{NIST}.% while the second is called \emph{Data Encapsulation Mechanism (DEM)}. Key feature is that the two parts are independent of one another. The framework was first introduced in a seminal work by Cramer and Shoup \cite{cs01}, along with the corresponding notions of security and an example of a scheme based on the DDH assumptions. In a subsequent work \cite{ISO}, Shoup presents a 
\section{Preliminaries}

\subsection{The McEliece Cryptosystem}

We consider here a more ``modern" version compared to R. J. McEliece's original cryptosystem~\cite{mceliece}. In the description that we use (Table~\ref{niedscheme}, Appendix~\ref{niedapp}), we consider families of codes to which is possible to associate an efficient decoding algorithm; we denote this with $\Ts{Decode}_\Delta$, where $\Delta$ is a description of the selected code that depends on the specific family considered. For instance, in the case of binary Goppa codes, the associated algorithm is Patterson's algorithm~\cite{patterson} and $\Delta$ is given by a Goppa polynomial $g(x)$ and its support $(\alpha_1,\dots,\alpha_n)$. For MDPC codes~\cite{MisoczkiTillichSendrierBarreto-2012}, decoding is given by Gallager's bit-flipping algorithm~\cite{gallager} and $\Delta$ is a sparse parity-check matrix $H$ for the code. Also, we denote with $\9W_{q,n,w}$ the set of words of $\9F_q^n$ with Hamming weight $w$.\5

The security of the scheme follows from the two following computational assumptions.

\begin{assumption}[Indistinguishability] \label{Nied-IND}The $k\times n$ matrix $G$ output by $\Gen$ is computationally indistinguishable from a same-size uniformly chosen matrix.
\end{assumption}%\vspace{-0.5cm}

\begin{assumption}[Decoding Hardness] \label{dechard}
Let $G$ be a generator matrix for an $[n,k]$ linear code $\7C$ over $\9F_q$ and $y$ a word of $\9F_q^n$. It is hard to find a codeword $c\in\7C$ with $d(c,y)\leq w$.
\end{assumption}

%
%\4It is easy to verify the consistency of the decryption process. In fact, we have $S^{-1}\psi=\8HP e^\T$ and since $P$ is a permutation matrix, the vector $P e^\T$ has still weight $w$. Decoding and then multiplying by $P^{-1}$ on the left returns the desired plaintext.\\

\4Assumption~\ref{dechard} is also known as the General Decoding Problem (GDP), which was proved to be NP-complete in \cite{np}, and it is believed to be hard on average, and not just on the worst-case instances (see for example Sendrier~\cite{sendconj}).

\subsection{Encapsulation Mechanisms and the Hybrid Framework}
A Key Encapsulation Mechanism (KEM) is essentially a Public-Key Encryption scheme (PKE), with the exception that the encryption algorithm takes no input apart from the public key, and returns a pair $(K,\psi_0)$. The string $K$ has fixed length $\ell_K$, specified by the KEM, and $\psi_0$ is an ``encryption" of $K$ in the sense that $\Dec_\sk(\psi_0)=K$. The key $K$ produced by the KEM is then passed on to a Data Encapsulation Mechanism (DEM), which is in charge of encrypting the actual message. The formulation of a DEM, that normally comprises additional tools for security such as Message Authentication Codes (MAC), is outside the scope of this paper, and we refer the reader to~\cite{cs01} for more details.\5

%Formally, a KEM consists of three algorithms: a \emph{Key Generation} algorithm $\Gen$, an \emph{Encapsulation} algorithm $\Enc$ that takes as input a public key $\pk$ and returns a key/ciphertext pair $(K,\psi_0)$, and a \emph{Decapsulation} algorithm $\Dec$ that takes as input a private key $\sk$ and a ciphertext $\psi_0$ and outputs either a key $K$ or the failure symbol $\bot$. 
A KEM is required to be \emph{sound} for at least all but a negligible portion of public key/private key pairs, that is, if $\Enc_\pk(\ )=(K,\psi_0)$ then $\Dec_\sk(\psi_0)=K$ with overwhelming probability.

%Formally, a KEM consists of three algorithms, $\Gen,\Enc,\Dec$, which we present in Table~\ref{KEMdef}. A KEM is required to be \emph{sound} for at least all but a negligible portion of public key/private key pairs, that is, if $\Enc_\pk(\ )=(K,\psi_0)$ then $\Dec_\sk(\psi_0)=K$ with overwhelming probability.
%\vspace{-0.5cm}
%\pagebreak

%\renewcommand{\arraystretch}{1.5}
%\begin{table}[h!]\small
%\caption{Key Encapsulation Mechanism.}\label{KEMdef}
%\begin{tabular}{lp{11cm}}
%\hline
%%\multirow{2}{*}{$\Ts K$}  & $\Ts K_\Ts {publ}$ the public key space.\\
%%& $\Ts K_\Ts {priv}$ the private key space.\\
% %$\Ts P$ & The set of messages to be encrypted, or \emph{plaintext space}.\\
% %$\Ts C$ & The set of the messages transmitted over the channel, or \emph{ciphertext space}.\\
% $\Gen$ & A probabilistic key generation algorithm that takes as input a security parameter $1^\lambda$ and outputs a public key $\pk$ and a private key $\sk$.\\
% $\Enc$ & A probabilistic encryption algorithm that receives as input a public key $\pk$ and returns a key/ciphertext pair $(K,\psi_0)$.\\
% $\Dec$ & A deterministic decryption algorithm that receives as input a private key $\sk$ and a ciphertext $\psi_0$ and outputs either a key $K$ or the failure symbol $\bot$.\\
% \hline
%\end{tabular}
%\end{table}
%

\medskip

The security notions for a KEM are similar to the corresponding ones for PKE schemes. The one we are mainly interested in (representing the highest level of security) is IND-CCA, which we describe below.\5

\begin{definition}\label{cca2KEM}
The adaptive Chosen-Ciphertext Attack game for a KEM proceeds as follows:
\begin{enumerate}\addtolength{\itemsep}{0.1\baselineskip}
\item  Query a key generation oracle to obtain a public key $\pk$.
\item  Make a sequence of calls to a decryption oracle, submitting any string $\psi_0$ of the proper length. The oracle will respond with $\Dec^\KEM_{\sk}(\psi_0)$.
\item  Query an encryption oracle. The oracle runs $\Enc^\KEM_\pk$ to generate a pair $(\tilde{K},\tilde{\psi_0})$, then chooses a random $b\in\{0,1\}$ and replies with the ``challenge" ciphertext $(K^*,\tilde{\psi_0})$ where $K^*=\tilde{K}$ if $b=1$ or $K^*$ is a random string of length $\ell_K$ otherwise.
\item  Keep performing decryption queries. If the submitted ciphertext is $\psi^*_0$, the oracle will return $\bot$.
\item  Output $b^*\in\{0,1\}$.
\end{enumerate}
The adversary succeeds if $b^*=b$. More precisely, we define the \emph{advantage} of $\7A$ against KEM as
\begin{equation}
\Adv_\KEM(\7A,\lambda)=\Big|\pr[b^*=b]-\frac{1}{2}\Big|.
\end{equation}
%\footnote{The adversary is free to choose this string in any arbitrary way, and not necessarily using the encryption algorithm.}

\4We say that a KEM is secure if the advantage $\Adv_\KEM$ of any polynomial-time adversary $\7A$ in the above CCA attack model is negligible.
\end{definition}

It has then been proved that, given a CCA adversary $\7A$ for the hybrid scheme (HY), there exist an adversary $\7A_1$ for KEM and an adversary $\7A_2$ for DEM running in roughly the same time as $\7A$, such that for any choice of the security parameter $\lambda$ we have $\Adv_\HY(\7A,\lambda)\leq \Adv'_\KEM(\7A_1,\lambda)+\Adv_\DEM(\7A_2,\lambda)$. See Cramer and Shoup \cite[Th. 5]{cs01} for a complete proof.

%\subsection{Other Cryptographic Tools}
%
%In this section we introduce another cryptographic tool that we need for our construction.
%
%\begin{definition}
%A \emph{Key Derivation Function (KDF)} is a function that takes as input a string $ x$ of arbitrary length and an integer $\ell\geq0$ and outputs a bit string of length $\ell$. 
%\end{definition}
%
%\4A KDF is modelled as a random oracle, and it satisfies the \emph{entropy smoothing} property, that is, if $ x$ is chosen at random from a high entropy distribution, the output of KDF should be computationally indistinguishable from a random length-$\ell$ bit string.\5
%
%Intuitively, a good choice for a KDF could be a hash function with a variable (arbitrary) length output, such as the new SHA-3, Keccak~\cite{keccak}.%\pagebreak

%\begin{definition}
%A \emph{Message Authentication Code (MAC)} is an algorithm that produces a short piece of information \emph{(tag)} used to authenticate a message. A MAC is defined by a function $\Ts{Ev}$ that takes as input a key $K$ of length $\ell_\Ts{MAC}$ and an arbitrary string $T$ and returns a tag to be appended to the message, that is, a string $\tau$ of fixed length $\ell_\Ts{TAG}$.
%\end{definition}
%
%\4Informally, a MAC is similar to a signature scheme, with the difference that the scheme makes use of private keys both for evaluation and verification; in this sense, it could be seen as a ``symmetric encryption equivalent" of a signature scheme. The usual desired security requirement is existential unforgeability under chosen message attacks (see Appendix~\ref{digsig}).

\section{The New KEM Construction}\label{nkem}
%\subsection{The KEM Construction}
The KEM we present here follows closely the McEliece framework, and is thus based on the hardness of GDP. Note that, compared to the original PKE, a slight modification is introduced in the decryption process. As we will see later, this is necessary for the proof of security. The ephemeral key $K$ is obtained via a Key Derivation Function $\KDF$ (see Appendix~\ref{appdef}).%Although unusual, this particular formulation still satisfies the requirements of a KEM.%\pagebreak

\begin{table}[h!]\small
\caption{The McEliece KEM.}
\begin{tabular}{lp{11cm}}\label{NiedKEM}
\Ts {Setup}& Fix public system parameters $q,n,k,w\in \9N$, then choose a family $\7F$ of $w$-error-correcting $[n,k]$ linear codes over $\9F_q$.\\
\hline
$\Gen$ & Choose a code $\7C\in\7F$ with code description $\Delta$ and compute a generator matrix $G$. Generate a random $s\rand\9F_q^k$. Public key is $G$ and private key is $(\Delta,s)$.\\
 $\Enc$ & On input a public key $G$ choose random words $x\in\9F_q^k$ and $ e\in\9W_{q,n,w}$, then compute $K=\KDF(x|| e,\ell_K)$, $\psi_0=xG+e$ and return the key/ciphertext pair $(K,\psi_0)$.\\
 $\Dec$ & On input a private key $\Delta$ and a ciphertext $\psi_0$, compute $\Ts{Decode}_\Delta(\psi_0)$. If the decoding succeeds, use its output $(x,e)$ to compute $K=\KDF(x|| e,\ell_K)$. Otherwise, set $K=\KDF(s||\psi_0,\ell_K)$. Return $K$.\\
 \hline
\end{tabular}
\end{table}%\vspace{-0.3cm}
%\footnotetext{A natural suggestion is for example to set $K=\KDF(\psi_0,\ell_K)$.}

\4If the ciphertext is correctly formed, decoding will always succeed, hence the KEM is perfectly sound. Furthermore, it is possible to show that, even if with this formulation $\Dec^\KEM$ never fails, there is no integrity loss in the hybrid encryption scheme thanks to the check given by the MAC.\\
\4We prove the security of the KEM in the following theorem.

\begin{theorem}\label{KEMsec}
Let $\7A$ be an adversary in the random oracle model for the Niederreiter KEM as in Definition~\ref{cca2KEM}. Let $\theta$ be the running time of $\7A$, $n_\KDF$ and $n_\Dec$ be two bounds on, respectively, the total number of random oracle queries and the total number of decryption queries performed by $\7A$, and set $N=q^k\cdot |\9W_{q,n,w}|$. Then there exists an adversary $\7A'$ for GDP such that $\Adv_\Ts{KEM}(\7A,\lambda)\leq\Adv_\Ts{GDP}(\7A',\lambda)+n_\Dec/N$. The running time of $\7A'$ will be approximately equal to $\theta$ plus the cost of $n_\KDF$ matrix-vector multiplications and some table lookups.
\end{theorem}

\begin{proof}

%\subsection{Proof of security}
We replace $\KDF$ with a random oracle $\7H$ mapping elements of the form $(x,e)\in\9F_q^k\times\9W_{q,n,w}$ to bit strings of length $\ell_K$. To prove our claim, we proceed as follows. Let's call $\Ts G_0$ the original attack game played by $\7A$, and $\Ts S_0$ the event that $\7A$ succeeds in game $\Ts G_0$. We define a new game $\Ts G_1$ which is identical to $\Ts G_0$ except that the game is halted if the challenge ciphertext $\psi^*_0=x^*G+e^*$ obtained when querying the encryption oracle had been previously submitted to the decryption oracle: we call this event $\Ts F_1$. Since the number of valid ciphertexts is $N$, we have $\pr[\Ts F_1]\leq n_\Dec/N$. It follows that $\Big|\pr[\Ts S_0]-\pr[\Ts S_1]\Big|\leq n_\Dec/N$, where $\Ts S_1$ is the event that $\7A$ succeeds in game $\Ts G_1$. Next, we define game $\Ts G_2$ which is identical to $\Ts G_1$ except that we generate the challenge ciphertext $\psi^*_0$ at the beginning of the game, and we halt if $\7A$ ever queries $\7H$ at $(x^* || e^*)$: we call this event $\Ts F_2$. By construction, since $\7H(x^* || e^*)$ is undefined, it is not possible to tell whether $K^*=K$, thus we have $\pr[\Ts S_2]=1/2$, where $\Ts S_2$ is the event that $\7A$ succeeds in game $\Ts G_2$. We obtain that $\Big|\pr[\Ts S_1]-\pr[\Ts S_2]\Big|\leq \pr[\Ts F_2]$ and we just need to bound $\pr[\Ts F_2]$.\\
We now construct an adversary $\7A'$ against GDP. $\7A'$ interacts with $\7A$ and is able to simulate the random oracle and the decryption oracle with the help of two tables $\Ts T_1$ and $\Ts T_2$, initially empty, as described below.\5

\4{\bf Key Generation}: On input the instance $(G, y^*,w)$ of GDP, return the public key $\pk=G$.\5

\4{\bf Challenge queries}: When $\7A$ asks for the challenge ciphertext:
\begin{enumerate}\addtolength{\itemsep}{0.4\baselineskip}
\item Generate a random string $K^*$ of length $\ell_K$.
\item Set $\psi^*_0= y^*$.
\item Return the pair $(K^*,\psi^*_0)$.
\end{enumerate}%\smallskip
%\pagebreak

\4{\bf Random oracle queries}: Upon $\7A$'s random oracle query $(x,e)\in\9F_q^k\times\9W_{q,n,w}$:
\begin{enumerate}\addtolength{\itemsep}{0.3\baselineskip}
\item Look up $(x,e)$ in $\Ts T_1$. If $(x,e,y,K)$ is in $\Ts T_1$ for some $y$ and $K$, return $K$.
\item Compute $y=xG+e$.
\item If $y=y^*$ then $\7A'$ outputs $c=xG$ and the game ends.
\item Look up $y$ in $\Ts T_2$. If $(y,K)$ is in $\Ts T_2$ for some $K$ (i.e. the decryption oracle has been evaluated at $y$), return $K$.
\item Set $K$ to be a random string of length $\ell_K$ and place $(x,e,y,K)$ in table $\Ts T_1$.
\item Return $K$.
\end{enumerate}

\4{\bf Decryption queries}: Upon $\7A$'s decryption query $y\in\9F_q^{n}$: 
\begin{enumerate}\addtolength{\itemsep}{0.3\baselineskip}
\item Look up $y$ in $\Ts T_2$. If $(y,K)$ is in $\Ts T_2$ for some $K$, return $K$.
\item Look up $y$ in $\Ts T_1$. If $(x,e,y,K)$ is in $\Ts T_1$ for some $x,e$ and $K$ (i.e. the random oracle has been evaluated at $(x,e)$ such that $y=xG+e$), return $K$.
\item Generate a random string $K$ of length $\ell_K$ and place the pair $(y,K)$ in $\Ts T_2$.
\item Return $K$.
\end{enumerate}%\pagebreak

\4Note that, in both random oracle and decryption queries, we added the initial steps to guarantee the integrity of the simulation, that is, if the same value is queried more than once, the same output is returned. A fundamental issue is that it is impossible for the simulator to determine if a word is decodable or not.  If the decryption algorithm returned $\bot$ if and only if a word was not decodable, then it would be impossible to simulate decryption properly.  We have resolved this problem by insisting that the KEM decryption algorithm always outputs a hash value. With this formulation, the simulation is flawless and $\7A'$ outputs a solution to the GDP instance with probability equal to $\pr[\Ts F_2]$.\qed
\end{proof}
\vspace{-0.5cm}

\section{Conclusions}\label{kemconcl}

In this paper, we have introduced a key encapsulation method based on the McEliece cryptosystem. This novel approach enjoys a simple construction and a tight security proof as for the case of the Niederreiter KEM presented in ~\cite{Edoardo-2013}. We believe that our new construction will offer an important alternative while designing quantum-secure cryptographic primitives.

\bibliographystyle{plain-perso}
\bibliography{biblio2}

 \begin{appendix}
 \section{The McEliece Cryptosystem}\label{niedapp}

%Due to space limitations, we leave a detailed description to Appendix~\ref{niedapp}.\\
\begin{table}[h!]%\small
\caption{The McEliece cryptosystem.}
\begin{tabular}{lp{11cm}}\label{niedscheme}
\Ts {Setup}& Fix public system parameters $q,n,k,w\in \9N$, then choose a family $\7F$ of $w$-error-correcting $[n,k]$ linear codes over $\9F_q$.\\
\hline
\multirow{2}{*}{$\Ts K$}  & $\Ts K_\Ts {publ}$ the set of $k\times n$ matrices over $\9F_q$.\\
& $\Ts K_\Ts {priv}$  the set of code descriptions for $\7F$.\\
$\Ts P$ & The vector space $\9F_q^{k}$.\\
$\Ts C$ & The vector space $\9F_q^n$.\\
$\Gen$ & Generate at random a code $\7C\in\7F$ given by its code description $\Delta$ and compute a public\footnotemark\ generator matrix $G$. Publish the public key $G\in \Ts K_\Ts {publ}$ and store the private key $\Delta\in \Ts K_\Ts {priv}$.\\
$\Enc$ & On input a public key $G\in\Ts K_\Ts {publ}$ and a plaintext $\phi=x\in\Ts P$, choose a random error vector $e\in\9W_{q,n,w}$, then compute $y=xG+e$ and return the ciphertext $\psi=y \in \Ts C$.\\
$\Dec$ & On input the private key $\Delta\in\Ts K_\Ts {priv}$ and a ciphertext $\psi\in\Ts C$, compute $\Ts{Decode}_\Delta(\psi)$. If the decoding succeeds, return its output $\phi=x$. Otherwise, output $\bot$.\\
\hline
\end{tabular}
\end{table}

\footnotetext{While the original version proposes to use scrambling matrices $S$ and $P$ (see~\cite{nied} for details), this is not necessary and alternative methods can be used, depending on the chosen code family.}

%\section{KEM Security Definition}

\section{Other Cryptographic Tools}\label{appdef}

In this section we introduce another cryptographic tool that we need for our construction.

\begin{definition}
A \emph{Key Derivation Function (KDF)} is a function that takes as input a string $ x$ of arbitrary length and an integer $\ell\geq0$ and outputs a bit string of length $\ell$. 
\end{definition}

\4A KDF is modelled as a random oracle, and it satisfies the \emph{entropy smoothing} property, that is, if $ x$ is chosen at random from a high entropy distribution, the output of KDF should be computationally indistinguishable from a random length-$\ell$ bit string.\5

Intuitively, a good choice for a KDF could be a hash function with a variable (arbitrary) length output, such as the new SHA-3, Keccak~\cite{keccak}.%\pagebreak

 \end{appendix}

\end{document}